\documentclass[11pt]{article}
\usepackage[affil-it]{authblk}
\usepackage{amsfonts}
\usepackage{latexsym} 
\usepackage{epsfig}
\usepackage{amssymb}
\usepackage{amsmath}
\usepackage{verbatim}

\newtheorem{theorem}{Theorem}[section] 
\newtheorem{corollary}[theorem]{Corollary} 
\newtheorem{prop}[theorem]{Proposition} 
 
\newtheorem{lemma}[theorem]{Lemma}

\newenvironment{proof}{{\bf Proof:\ }}{\hfill $\Box$\linebreak\vskip2mm}

\def\tol{{\sf tol}}
  
\let\meet=\wedge  
\let\join=\vee  
\let\tm=\times

\def\oa{\ov a}  
\def\ob{\ov b}

\let\ov=\overline  
  
\def\zA{{\mathbb A}}  
\def\zB{{\mathbb B}}

\def\cC{{\cal C}}

\def\cG{{\cal G}}

\def\cP{{\cal P}}

\def\gA{{\mathfrak A}}

\def\ba{{\bf a}}  
\def\bb{{\bf b}}  
\def\bc{{\bf c}}  
\def\bd{{\bf d}}  
\def\be{{\bf e}}

\def\bs{{\bf s}}

\def\bx{{\bf x}}

\def\CSP{{\rm CSP}}

\def\rel{{R}}
\def\relo{{S}}

\let\sse=\subseteq  
\def\ang#1{\langle #1 \rangle}  
\def\vc#1#2{#1 _1\zd #1 _{#2}}  
\def\tc#1#2{#1 _1\td #1 _{#2}}  
\def\zd{,\ldots,}  
\def\td{\tm\ldots\tm}
  
\def\red#1{\vrule height7pt depth4pt width.4pt  
     \lower3pt\hbox{$\scriptstyle #1$}}  
\def\fac#1{/\lower4pt\hbox{$#1$}}

\def\cl#1#2{\arraycolsep0pt 
\left(\begin{array}{c} #1\\ #2 \end{array}\right)}

\def\pr{{\rm pr}}

\def\lev{{\sf lev}}

\def\summ{{\sf sum}}

\def\lb{$\linebreak$}

\let\al=\alpha  
  
\let\dl=\delta

\let\vf=\varphi

\let\Gm=\Gamma

\font\tengoth=eufm10 scaled 1200  
\font\sixgoth=eufm6  
  
\font\tenbur=msbm10  
\font\eightbur=msbm8  
  
\font\twelvebur=msbm10 scaled 1200  
\mathchardef\nat="0B4E  
\mathchardef\eps="0D3F  
  
\begin{document} 
 

\title{Conservative constraint satisfaction re-revisited} 
\author{Andrei A.\ Bulatov\thanks{This research is supported by an NSERC Discovery Grant.}
}
\affil{School of Computing Science, Simon Fraser University\\
Email: abulatov@sfu.ca}
\date{}
\maketitle

\begin{abstract}
Conservative constraint satisfaction problems (CSPs) constitute an important 
particular case of the general CSP, in which the allowed values of each variable
can be restricted in an arbitrary way. Problems of this type are well studied 
for graph homomorphisms. A dichotomy theorem characterizing 
conservative CSPs solvable in polynomial time and proving that the remaining
ones are NP-complete was proved by Bulatov in 2003.
Its proof, however, is quite long and technical. A shorter proof of this result 
based on the absorbing subuniverses technique was suggested by Barto in 2011. 
In this paper we give a short elementary prove of the dichotomy theorem for 
the conservative CSP. 
\end{abstract}

\section{Introduction}

In a constraint satisfaction problem (CSP) the aim is to find an
assignment of values to a given set of variables, subject to
specified constraints. The CSP is known to be
NP-complete in general. However, certain
restrictions on the form of the allowed constraints can lead to problems
solvable in polynomial time. Such restrictions are usually imposed by
specifying a constraint language,
that is, a set of relations that are allowed to be used as
constraints. A principal research direction aims
to distinguish those constraint languages that give rise to 
CSPs solvable in polynomial time from those that do not. 
The dichotomy conjecture \cite{Feder98:monotone} 
suggests that every constraint language gives rise to a CSP that
is either solvable in polynomial time or is NP-complete. The dichotomy
conjecture is confirmed in a variety of particular cases 
\cite{Barto11:conservative,Bulatov11:conservative,Bulatov02:3-element,%
Bulatov03:conservative,Bulatov06:3-element,Hell90:h-coloring,Schaefer78:complexity}, 
but the general problem remains open.

One of the important versions of the CSP is often referred to as the
conservative or list CSP. In a CSP of this type the set of values for 
each individual variable can be restricted
arbitrarily. Restrictions of this type can be studied by considering
those constraint languages which contain all possible unary constraints;
such languages are also called conservative. Conservative 
CSPs have been intensively studied for languages consisting of only one 
binary symmetric relation, that is, graphs; in this case CSP is 
equivalent to the graph homomorphism problem 
\cite{Feder98:list,Feder99:list,Feder99:bi-arc,Hell90:h-coloring,Kratochvil94:algorithmic}. 

In \cite{Bulatov11:conservative,Bulatov03:conservative} the dichotomy 
conjecture was confirmed for conservative CSPs. However, the proof 
given in \cite{Bulatov11:conservative,Bulatov03:conservative} 
is quite long and technical, which prompted attempts to find a simpler
argument. In \cite{Barto11:conservative} Barto gave a simpler proof 
using the absorbing 
subuniverses techniques. In the present paper we give another, more elementary,
proof that applies the reduction suggested in \cite{Maroti10:tree}.

As in the majority of dichotomy results the solution
algorithm and the proofs heavily use the algebraic approach to the CSP
developed in \cite{Bulatov03:multisorted,
Bulatov05:classifying,Jeavons97:closure,Jeavons98:algebraic}. 
This approach relates a constraint language to a collection of 
polymorphisms of the language, that is, operations on the same set 
that preserves all the relations from the language, and uses 
polymorphisms of specific 
types to identify constraint languages solvable in polynomial time.
For example, to characterize CSPs on a 2-element set solvable
in polynomial time \cite{Schaefer78:complexity} it suffices to consider only 
4 types of operations on a 2-element set: constant, semilattice 
(conjunction and disjunction), majority 
($(x\meet y)\join(y\meet z)\join(z\meet x)$), and affine ($x+y+z$). 
The same types of operations characterize the complexity of 
conservative CSPs, except that constant operations cannot be 
polymorphisms of conservative languages. In a simplified form 
the main result we prove is 
\begin{theorem}[\cite{Bulatov11:conservative,Bulatov03:conservative}]\label{the:main-intro}
Let $\Gm$ be a constraint language on a set $A$. The conservative
CSP using relations from $\Gm$ can be solved in polynomial time
if and only if for any 2-element subset $\{a,b\}\sse A$ there is an
operation $f$ on $A$, a polymorphism of $\Gm$, such that $f$ on 
$\{a,b\}$ is either a semilattice operation, or a majority operation, 
or an affine operation. Otherwise this CSP is NP-complete.
\end{theorem}

We give a new nearly complete proof of Theorem~\ref{the:main-intro}. The only 
statements we reuse in this paper is Proposition~\ref{pro:uniformity} 
that we borrow from \cite{Bulatov11:conservative} and the results of
Setion~\ref{sec:maroti}.

\section{Definitions and preliminaries}

\subsection{Constraint satisfaction problems and algebra}

By $[n]$ we denote the set $\{1\zd n\}$.
Let $\vc An$ be sets, any element of $\tc An$ is an ($n$-ary) 
tuple. Tuples will be denoted in boldface, say, $\ba$, and the $i$th 
component of $\ba$ will be referred to as $\ba[i]$. An $n$-ary relation over 
$\vc An$ is any set of tuples over these sets. For a set $I=\{\vc ik\}\sse[n]$, 
a tuple $\ba\in\tc An$, and a relation $\rel\sse\tc An$, by $\pr_I\ba$
we denote the tuple $(\ba[i_1]\zd\ba[i_k])$, the \emph{projection} of 
$\ba$ on $I$, and $\pr_I\rel=\{\pr_I\bb\mid\bb\in\rel\}$ denotes the 
projection of $\rel$ on $I$. Relation $\rel$ is said to be a \emph{subdirect 
product} of $\vc An$ if $\pr_i\rel=A_i$ for all $i\in[n]$. Let $I\sse[n]$.
For $\ba\in\pr_I\rel$ and $\bb\in\pr_{[n]-I}\rel$ by $(\ba,\bb)$ 
we denote the tuple $\bc$ such that $\bc[i]=\ba[i]$ if $i\in I$ and 
$\bc[i]=\bb[i]$ otherwise.

Let $\gA$ be a collection of finite sets (in this paper we assume 
$\gA$ to be finite as well). A \emph{constraint satisfaction problem} 
over $\gA$ is a triple $(V,\dl,\cC)$, where $V$ is a (finite) set of
\emph{variables}, $\dl$ is a \emph{domain function}, $\dl:V\to\gA$
assigning a domain of values to every variable, and $\cC$ is a set
of \emph{constraints}. Every constraint is a pair $\ang{\bs,\rel}$,
where $\bs=(\vc vk)$ is a sequence of variables from $V$ (possibly 
with repetitions) called the \emph{constraint scope}, and $\rel$ is 
a relation over $\dl(v_1)\td\dl(v_k)$ called the \emph{constraint 
relation}. A mapping $\vf:V\to\bigcup\gA$ that maps every variable
$v$ to its domain $\dl(v)$ is called a \emph{solution} if for every 
$\ang{\bs,\rel}\in\cC$ we have $\vf(\bs)\in\rel$.

Let $W\sse V$. A \emph{partial solution} of $\cP$ on $W$ is a
mapping $\vf:W\to\bigcup\gA$ such that for every constraint
$\ang{\bs,\rel}\in\cC$, $\bs=(\vc vk)$, we have $\vf(\bs')\in\pr_I\rel$, 
where $I=\{\vc i\ell\}$ is the set of indices $i_s$ from $[k]$ such that
$v_{i_s}\in W$, and $\bs'=(v_{i_1}\zd v_{i_\ell})$. The set of all partial 
solutions on  set $W$ is denoted by $\relo_W$. Problem $\cP$ 
is said to be \emph{3-minimal} if it contains a constraint $\ang{W,\relo_W}$
for every 3-element $W\sse V$, and for any $W_1,W_2\sse V$
such that $|W_1|=|W_2|=3$ and $|W_1\cap W_2|=2$, 
$\pr_{W_1\cap W_2}\relo_W=\pr_{W_1\cap W_2}\relo_{W_1}
\cap \pr_{W_1\cap W_2}\relo_{W_2}$. There are standard polynomial 
time \emph{propagation} algorithms (see, e.g.\ \cite{Dechter03:processing}) 
to convert any CSP to an equivalent, that is, having the same
solutions, 3-minimal CSP.

An introduction into universal algebra and the algebraic approach 
to CSP can be found in 
\cite{Burris81:universal,Bulatov03:multisorted,Bulatov05:classifying,Bulatov11:conservative}. 
Here we only mention several key points. For an algebra $\zA$ its
universe will be denoted by $A$. Let $\gA$
be a finite collection of finite similar algebras. For a basic or term 
operation $f$ of the class $\gA$ by $f^\zA$, $\zA\in\gA$, we
denote the interpretation of $f$ in $\zA$. Let $\vc\zA k\in\gA$. 
A relation 
$\rel\sse \tc Ak$ is a \emph{subalgebra} of the direct product $\tc\zA k$, 
denoted $\rel\le\tc\zA k$, if for any
basic operation $f$ (say, it is $n$-ary) of $\gA$ and any 
$\vc\ba n\in\rel$ the tuple $f(\vc\ba k)=(f^{\zA_1}(\ba_1[1]\zd
\ba_n[1])\zd f^{\zA_k}(\ba_1[k]\zd\ba_n[k]))$ belongs to $\rel$.
In this case $f$ is also said to be a \emph{polymorphism} of $\rel$.

By $\CSP(\gA)$ we
denote the class of CSP problems $\cP=(V,\dl,\cC)$ such that
$\dl(v)$ is the universe of one of the members of $\gA$, and
every constraint relation is a subalgebra of the direct product 
of the domain algebras. In this paper we assume that the 
algebras from $\gA$ satisfy certain requirements. 
An algebra is said to be \emph{conservative}
if every subset of its universe is a subalgebra. We only consider
classes of conservative algebras.  Also,  the class $\gA$ will be
assumed to be closed under subalgeras. That is, if $\zA\in\gA$ 
then every subalgebra of $\zA$ also belongs to $\gA$. By 
\cite{Bulatov03:multisorted,Bulatov05:classifying}, for any
finite $\gA$ the problem $\CSP(\gA)$ has the same complexity
as $\gA'$, where $\gA'$ is obtained from $\gA$ by adding all
the subalgebras of algebras from $\gA$. 
A \emph{unary polynomial} of an algebra 
$\zA$ is a mapping $p:\zA\to\zA$, for which there exists a term
operation $t(x,\vc yk)$ and elements $\vc ak\in\zA$ such that
$p(x)=t(x,\vc ak)$. Unary polynomial $p(x)$ is idempotent if
$p(p(x))=p(x)$. The \emph{retract} of $\zA$ via polynomial $p(x)$
is the algebra $p(\zA)$ with the universe $p(A)$, where $A$ is
the universe of $\zA$ and term operations $p(t)$, where
$t(\vc xk)$ is a term operation of $\zA$ and $p(t)(\vc ak)=p(t(\vc ak))$
for any $\vc ak\in p(A)$. We will additionally assume that class 
$\gA$ is closed under retracts. This however does not impose
any additional restrictions in the case of conservative algebras,
since, as is easily seen, every retract of a conservative algebra is a subalgebra.

A subalgebra $\zA$ of a direct product of algebras $\tc\zA n$ is said to be 
a \emph{subdirect product} if the universe of $\zA$ viewed as a relation
is a subdirect product of the universes of $\vc\zA n$. For a congruence $\al$ of 
algebra $\zA$ and element $a$ by $\zA\fac\al$ we denote the 
factor-algebra of $\zA$ and by $a^\al$ the block of $\al$ containing
$a$.

\subsection{Graphs, paths, and the three basic operations}

If $\gA$ is a class of conservative algebras closed under subalgebras, 
then every subalgebra $\zB$ of any $\zA\in\gA$ belongs to
$\gA$. Therefore, by \cite{Schaefer78:complexity}, if $\CSP(\gA)$ is 
polynomial time solvable then, for any 2-element
subalgebra $\zB$ of $\zA$ (we assume $\zB=\{0,1\}$), there exists a term
operation $f_\zB$ of $\gA$ such that $f^\zB_\zB$ is one of the operations
yielding the tractability of the CSP on a 2-element set:  $f^\zB_ \zB$ is either a
\emph{semilattice} (that is conjunction or disjunction) operation, or the
\emph{majority} operation $(x\vee y)\wedge(y\vee z)\wedge(z\vee x)$, or the
\emph{affine} operation $x-y+z\!\!\!\pmod2$. Note that the constant
operations are not in this list since $\Gm$ is conservative. In 
\cite{Bulatov03:conservative,Bulatov11:conservative} it was 
proved that this property is also sufficient for
the tractability of $\CSP(\gA)$.

\begin{theorem}[\cite{Bulatov03:conservative,Bulatov11:conservative}]\label{the:main}
Let $\gA$ be a finite class of conservative algebras. The problem
$\CSP(\gA)$ can be solved in polynomial time if and only if for any
$\zA\in\gA$ and any 2-element subalgebra $\zB$ of $\zA$ there is
a term operation $f_\zB$ of $\gA$ such that $f^\zB_\zB$ is either
semilattice, or majority, or affine. Otherwise $\CSP(\gA)$ is NP-complete.
\end{theorem}

Let $\gA$ be a finite class of conservative algebras closed under subalgebras 
that satisfies the conditions of Theorem~\ref{the:main}. For every $\zA\in\gA$, 
we consider the graph  
$\cG_\gA(\zA)$, an edge-labeled digraph with vertex set $A$. 
An edge $(a,b)$ exists and is labeled {\em semilattice}
if there is a term operation $f_{a,b}$ of $\gA$ such that
$f^\zA_{a,b}\red{\{a,b\}}$ is a semilattice operation with
$f^\zA_{a,b}(a,b)=f^\zA_{a,b}(b,a)=f^\zA_{a,b}(b,b)=b$, $f^\zA_{a,b}(a,a)=a$.  
Edges $(a,b),(b,a)$ exist and are labeled {\em majority} if neither
$(a,b)$ nor $(b,a)$ is semilattice and there is a term operation
$f_{a,b}$ such that $f^\zA_{a,b}\red{\{a,b\}}$ is a
majority operation. 
Edges $(a,b),(b,a)$ exist and are labeled {\em affine} if none of
them is semilattice or majority, and there is a term operation
$f_{a,b}$ such that $f^\zA_{a,b}\red{\{a,b\}}$ is an affine
operation. Thus, for each pair $a,b\in A$, either $(a,b)$ or $(b,a)$
is an edge of $\cG_\gA(\zA)$; if $(a,b)$ is a majority or affine edge then
$(b,a)$ is also an edge with the same label; while if $(a,b)$ is semilattice
then the edge $(b,a)$ may not exist. Since $\gA$ is usually fixed, 
we shall use $\cG(\zA)$ rather than $\cG_\gA(\zA)$. 
The operations of the form $f_{a,b}$ can be considerably unified. 
%
\begin{prop}\label{pro:uniformity}
There are term operations $f(x,y),g(x,y,z),h(x,y,z)$ of $\gA$ such
that, for every $\zA\in\gA$ and every two-element subset $B\sse A$, 
\begin{itemize}
\item
$f^\zA\red B$ is a semilattice operation whenever $B$ is semilattice, and
$f^\zA\red B(x,y)=x$ otherwise; 
\item
$g^\zA\red B$ is a majority operation if $B$ is majority,
$g^\zA\red B(x,y,z)=x$ if $B$ is affine, and $g^\zA\red
B(x,y,z)=f^\zA\red B(f^\zA\red B(x,y),z)$ if $B$ is semilattice;
\item
$h^\zA\red B$ is an affine operation if $B$ is affine,
$h^\zA\red B(x,y,z)=x$ if $B$ is majority, and $h^\zA\red
B(x,y,z)=f^\zA\red B(f^\zA\red B(x,y),z)$ if $B$ is semilattice.
\end{itemize}
There is also a term operation $p(x,y)$ such that $p^\zA\red B=f^\zA\red B$ if
$B$ is semilattice, $p^\zA\red B(x,y)=y$ if $B$ is majority, and $p^\zA\red B(x,y)=x$ if
$B$ is affine.
\end{prop}
%

Using Proposition~\ref{pro:uniformity} we may assume that all algebras 
in $\gA$ have only three basic operations. We will normally use 
$\cdot$ instead of $f$. Operation $\cdot$ acts non-symmetrically on semilattice edges. 
This means that every such edge $ab$ is oriented: $ab$ is oriented from 
$a$ to $b$ if $a\cdot b=b\cdot a=b$; in this case we also write $a\le b$. 
Therefore $\cG(\zA)$ is treated as a digraph, in which semilattice edges are 
oriented, while majority and affine ones are not.

For a relation $\rel\le\zA_1\tm\ldots\tm\zA_n$ a digraph $\cG(\rel)$ can be 
defined in a natural way: tuples $\ba,\bb\in\rel$ form a semilattice edge directed 
from $\ba$ to $\bb$ if $\ba[i]=\bb[i]$ or $\ba[i]\bb[i]$ is a semilattice edge 
directed from $\ba[i]$ to $\bb[i]$ for every $i\in[n]$; tuples $\ba,\bb$ form 
a majority edge if $\ba[i]=\bb[i]$ or $\ba[i]\bb[i]$ is majority for each $i\in [n]$; 
and $\ba,\bb$ form an affine edge, if $\ba[i]=\bb[i]$ or $\ba[i]\bb[i]$ is 
affine for every $i\in[n]$. As is easily seen, graph
$\cG(\rel)$ is usually not complete, but as we shall see it inherits many
properties of the graph $\cG(\zA)$ of a conservative algebra $\zA$.

A sequence of vertices $\vc\ba k$ of $\cG(\rel)$ is a \emph{path}
if every $\ba_i\ba_{i+1}$ is either a semilattice or affine edge.

\begin{lemma}\label{lem:path-extension}
Let $\rel\le\zA_1\tm\ldots\tm\zA_n$, $I\sse[n]$, and let $\vc\ba k$
be a path in $\cG(\pr_I\rel)$. There are $\vc\bb k\in\pr_{[n]-I}\rel$
such that $(\ba_1,\bb_1)\zd(\ba_k,\bb_k)$ is a path in $\rel$.
\end{lemma}

\begin{proof}
Observe that for any $\zA\in\gA$ and any $a,b\in\zA$, the edge $a\,p(b,a)$
is either semilattice or affine. Therefore, for any $\ba,\bb\in\rel$, 
the pair $\ba\bc$, where $\bc=\ba\cdot p(\bb,\ba)$, is semilattice, while
$\ba\bd$, where $\bd=p(\bb,\ba)\cdot\ba$, is affine.

Take any $\vc\bc k\in\pr_{[n]-I}\rel$ such that $(\ba_i,\bc_i)\in\rel$ and
define $\vc\bb k$ as follows: $\bb_1=\bc_1$, if $\ba_i\ba_{i+1}$ 
is semilattice then $\bb_{i+1}=\bb_i\cdot p(\bc_{i+1},\bb_i)$, and if
$\ba_i\ba_{i+1}$ is affine then $\bb_{i+1}=p(\bc_{i+1},\bb_i)\cdot\bb_i$.
As is easily seen, $\vc\bb k$ satisfy the conditions of the lemma.
\end{proof}

A set $S\sse\rel$ is said to be \emph{connected} if there is a path
from every element in $S$ to every other element in $S$.

\section{Properties of labeled graph of algebras}

\subsection{As-components, linked relations, and connectivity}

Let $\zA\in\gA$ be a conservative algebra. A set $B\sse A$ is called an
\emph{as-component} (for affine-semilattice) if for any $a\in A$ and $b\in A-B$ the 
edge $ba$ is either majority or semilattice directed from $b$ to $a$, see 
Fig.~\ref{fig:as-comp}. Since as-components are defined in terms of 
the graph $\cG(\zA)$, this definition can be naturally generalized to 
as-components of relations.

\begin{figure}[ht]
\centerline{\includegraphics[totalheight=5cm,keepaspectratio]{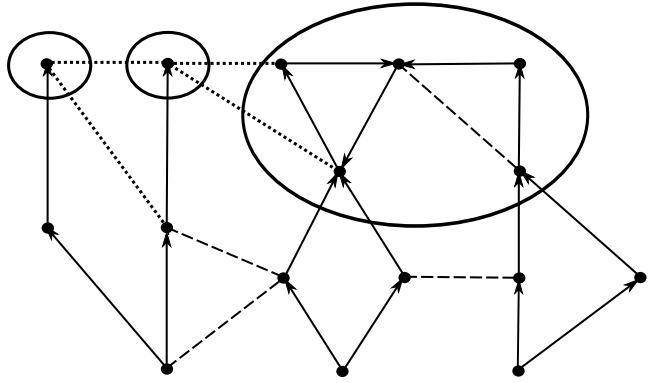}}
\caption{As-components. Solid lines represent semilattice edges, dashed lines 
represent affine edges, dotted lines represent majority edges; the edges that
are not shown are majority; as-components are encircled.}
\label{fig:as-comp}
\end{figure}

Let $\rel\le\zA\tm\zB$, where $\zA,\zB$ are subdirect products of conservative 
algebras. By $\tol_1(\rel)$ we denote the congruence of $\zA$ defined as  
the transitive close of the set $\{(a,b)\in\zA^2\mid \text{ there is } c\in\zB 
\text{ with } (a,c),(b,c)\in\rel\}$. Then $\tol_2(\rel)$ denotes the congruence on
$\zB$ defined in a similar way. Relation $\rel$ is said to be \emph{linked} if 
$\tol_1(\rel),\tol_2(\rel)$ are total relations.

\begin{lemma}\label{lem:as-subdirect}
Let $\rel\le\zA\tm\zB$, and $A',B'$ be as-components of 
$\zA,\zB$, respectively, such that $\rel'=\rel\cap(A'\tm B')\ne\eps$. Then 
$\rel'$ is a subdirect product of $A',B'$.
\end{lemma}

\begin{proof}
Let $A''=\pr_1\rel'\sse A'$. If $A''\ne A'$, there are $a\in A''$ and $a'\in A'-A''$
such that $a\le a'$ or $aa'$ is an affine edge. Take $(a,b),(a',b')\in\rel$ with
$b\in B'$. As is easily seen, $p\left(\cl{a'}{b'},\cl ab\right)=\cl{a'}b\in\rel$, since
$b'\not\in B'$, implying $a'\in A''$.
\end{proof} 

\begin{lemma}\label{lem:buket}
Let $\rel\le\zA\tm\zB$, and let $A',B'$ be as-components of 
$\zA,\zB$, respectively, such that there is $a\in A'$ with $\{a\}\tm B'\sse\rel$.
Then $A'\tm B'\sse\rel$.
\end{lemma}

\begin{proof}
By Lemma~\ref{lem:as-subdirect} $\rel\cap(A'\tm B')$ is a subdirect product
of $A',B'$. Therefore, if $A'\tm B'\not\sse\rel$ there are $b,c\in A'$, $d,e\in B'$ such that
$(b,d),(b,e),\lb
(c,e)\in\rel$, $(c,d)\not\in\rel$, $b\le c$ or $bc$ is affine, and
$e\le d$ or $ed$ is affine. If at least one of these two edges is not affine, we
have $\cl cd\in\left\{p\left(\cl bd,\cl ce\right),\right.\lb\left. p\left(\cl ce,\cl bd\right)\right\}$.
If both edges are affine then $\cl cd=h\left(\cl bd,\cl ce,\cl cd\right)$, 
a contradiction.
\end{proof}

\begin{lemma}\label{lem:linked-rectangularity}
Let $\zA,\zB$ be subdirect products of conservative algebras and 
let $\rel\le\zA\tm\zB$ be subdirect and linked. Let also $A',B'$ be 
as-components of $\zA,\zB$, respectively, such that $\rel\cap(A'\tm B')
\ne\eps$. Then $A'\tm B'\sse\rel$.
\end{lemma}

\begin{proof} 
We prove by induction on the size of $\zA,\zB$.
The base case of induction, when $|\zA|=1$ or $|\zB|=1$, is obvious.

Take any $\bb\in \zA$ and construct a sequence of subalgebras
$\vc Bk$ such that $B_i\sse \zA$ if $i$ is odd and $B_i\sse \zB$
if $i$ is even, as follows: $B_1=\{\bb\}$, $B_i=\rel[B_{i-1}]=
\{\bd\mid(\bc,\bd)\in\rel\text{ for some }\bc\in B_{i-1}\}$ if $i$ 
is odd, and $B_i=\rel^{-1}[B_{i-1}]=\{\bc\mid(\bc,\bd)\in\rel
\text{ for some }\bd\in B_{i-1}\}$ otherwise. 
By construction
for each $i\le k$ the relation $\rel_i=\rel'\cap(B_i\tm B_{i+1})$
(or $\rel_i=\rel'\cap(B_{i+1}\tm B_i)$) is linked.
Let $k$ be the maximal with $B_k\subset \zA$ or $B_k\subset \zB$. 
Without loss of generality we assume $B_k\subset \zA$. Set $\zA''= B_k$.
Thus there exists $\zA''\subset \zA$ such that $\rel'=\rel\cap(\zA''\tm \zB)
\sse\zA''\tm\zB$ is linked and subdirect. Choose a
minimal subalgebra $\zA''$ with this property. We show that there is 
$\ba\in \zA''$ such that $\{\ba\}\tm B'\sse\rel$.

If there is an as-component $C$ of $\zA''$ such that $\rel\cap(C\tm B')\ne\eps$
then $C\tm B'\sse\rel$ by induction hypothesis, and the claim follows. 
Let $D=\rel^{-1}[B']$. If $D$ contains no elements from an 
as-component, there are $\bb\in D$ and $\bc\in\zA-D$ such that $\bb\bc$ 
is a semilattice or affine edge. Take $\bb'\in B'$ and $\bc'\in\zB$
such that $(\bb,\bb'),(\bc,\bc')\in\rel$. Let
$$
\cl{\bc}{\bd}=p\left(\cl{\bc}{\bc'},\cl{\bb}{\bb'}\right), 
\quad\text{and}\quad
\cl{\bc''}{\bd'}=\cl{\bb}{\bb'}\cdot\cl{\bc}{\bd},\quad \bc''\in\{\bb,\bc\}. 
$$
Suppose $\zB\le\zA_1\tm\ldots\tm\zA_k$. Then for any $i\in[k]$ the pair
$\bb'[i]\bd[i]$ is a semilattice or affine edge. If there is no semilattice 
edge of this form then $\bb'\bd$ is an affine edge, implying $\bd\in B'$,
and $\bc\in D$, a contradiction. Otherwise $\bb'\bd'$ is a semilattice 
edge and $\bd'\bd$ is an affine one, hence $\bd\in B'$, a contradiction 
again. 

Let now $\ba\in\zA''$ be such that $\{\ba\}\tm B'\sse\rel$. If $\ba'\in A'$,
we are done. Otherwise take any $\bb'\in A'$ with $\rel[\bb]\cap B'\ne\eps$, 
and set $\bb=p(\ba,\bb')$. As before, it is easy to see that $\bb\in A'$.
Moreover, $p(\ba,\bb)=\bb$. Let also $B''=B'\cap\rel[\bb]$. 
If $B''\ne B'$, there is $\bc\in B'- B''$
and $\bd\in B''$ such that $\bd\bc$ is a semilattice or affine edge. Then 
$$ 
p\left(\cl{\ba}{\bc},\cl{\bb}{\bd}\right)=\cl{\bb}{\bc}, 
$$ 
a contradiction. Thus $\ba$ can be chosen from
$A'$. The proof is now completed by Lemma~\ref{lem:buket}.
\end{proof}

\begin{lemma}\label{lem:connectivity}
Let $\rel\le\zA_1\tm\ldots\tm\zA_n$ for conservative algebras $\vc\zA n$, 
and let $A'_i$ be an as-component
of $\zA_i$ for $i\in[n]$, such that $(\vc an)\in\rel$ for some $a_i\in A'_i$, $i\in[n]$.
Then $\rel'=\rel\cap(A'_1\tm\ldots\tm A'_n)$ is a subdirect product of the 
$A'_i$ and $\rel'$ is an as-component of $\rel$.
\end{lemma}

\begin{proof}
Let us first suppose that $\vc\zA n$ are simple. We prove the result by
induction on $n$. The trivial case $n=1$ gives the base case of induction.
Otherwise, we consider $\rel$ as a binary relation, a subdirect product of
$\zA=\pr_{[n-1]}\rel$ and $\zA_n$. 

Let $\ba,\bb\in\rel'$, $\ba'=\pr_{[n-1]}\ba$, $\bb'=\pr_{[n-1]}\bb$, 
and $a=\ba[n]$, $b=\bb[n]$. By the induction hypothesis there is a path
$\ba'=\ba'_1,\ba'_2\zd\ba'_k=\bb'$ in $\pr_{[n-1]}\rel'$. There are two 
cases. If $\rel$ is linked (as a subdirect product of $\zA\tm\zA_n$, then
$\pr_{[n-1]}\rel'\tm A'_n\sse\rel'$. Otherwise, 
as $\rel$ is not linked and $\zA_n$ is simple, for every $\bc'\in\zA$ there
is a unique $c\in\zA_n$ such that $(\bc',c)\in\rel$. In particular, 
there are unique $\vc ak$ such that $(\ba_i,a_i)\in\rel$. It is not
hard to see that if $\ba_i\ba_{i+1}$ is a semilattice (affine) edge, so is $a_ia_{i+1}$,
because otherwise $\ba_i$ or $\ba_{i+1}$ has more than one extension. 
Thus $(\ba_1,a_1)\zd(\ba_k,a_k)$ is a path from $\ba$ to $\bb$.

Suppose that not all of the algebras $\vc\zA n$ are simple.
We prove the lemma by induction on the number of non-simple
factors and their size.

We start with a couple of simple observations. If $\zA$ is a conservative 
algebra and $\al$ is its congruence, then $\zA\fac\al$ is also 
a conservative algebra. Moreover, if $\oa\ob$, $\oa,\ob\in\zA\fac\al$ 
is a semilattice (majority, affine)
edge of $\zA\fac\al$ then for any $a\in\oa, b\in\ob$ the edge $ab$
is also semilattice (respectively, majority, affine). It follows immediately from the 
observation that if $m\in\{f,g,h\}$ then $\{a,b\}$ is closed under
$m$, and $m(x,y,z)=a$ for $x,y,z\in\{a,b\}$ if and only if 
$m(x^\al,y^\al,z^\al)=\oa$. 

Suppose that $\zA_n$ is not simple and $\al$ is its maximal congruence. 
From the observation above it follows that $A''_n=\{a^\al\mid a\in A'_n\}$
is an as-component of $\zA_n\fac\al$.
Consider the relation 
$\relo=\{(\vc a{n-1},a_n^\al)\mid (\vc an)\in\rel\}$. By the 
induction hypothesis $\relo'=\{(\vc a{n-1},a_n^\al)\mid 
(\vc an)\in\rel'\}$ is connected and is an as-component of $\relo$.
Take $\ba,\bb\in\rel'$ and let $\ba',\bb'$ be the corresponding 
tuples from $\relo'$. Then there is a path $\ba'=\ba'_1,\ba'_2
\zd \ba'_k=\bb'$. For each $i\in[k]$ pick a tuple $\ba_i\in\rel'$
such that $\ba_i[n]\in\ba'_i[n]$. By the observation above, if
$\ba'_i\ba'_{i+1}$ is a semilattice (affine) edge, so is $\ba_i[n]\ba_{i+1}[n]$, 
and $\ba_i\ba_{i+1}$, as well. The sequence $\vc\ba k$ is a path from 
$\ba$ to $\bb$.
\end{proof}

\subsection{Rectangularity}\label{sec:rectangularity}

Let $\rel\le\zA_1\tm\ldots\tm\zA_n$ and let $A'_i\sse\zA_i$, $A'_j\sse\zA_j$
be as-components of $\zA_i$, $\zA_j$, respectively. Positions $i$ and $j$
are said to be \emph{$A'_i,A'_j$-related} if $\ba[i]\in A'_i$
if and only if $\ba[j]\in A'_j$, for any $\ba\in\rel$. A set $I\sse[n]$ is called 
a \emph{strand} with respect to as-component $\vc {A'}n$ of $\vc\zA n$, 
respectively, if it is maximal such that any $i,j\in I$ are $A'_i,A'_j$-related. 
As is easily seen, the strands with respect $\vc {A'}n$ form a partition of $[n]$.

\begin{lemma}\label{lem:rectangularity}
Let $\rel\le\zA_1\tm\ldots\tm\zA_n$ and let $\vc {A'}n$ be as-components 
of $\vc\zA n$, respectively, such that $\rel\cap(A'_1\tm\ldots\tm A'_n)\ne\eps$.
Let also $\vc Ik$ be the partition of $[n]$ into strands with respect to
$\vc {A'}n$ and $\rel_i=\pr_{I_j}\rel\cap\prod_{\ell\in I_j}A'_\ell$.
Then $\rel_1\tm\ldots\tm\rel_k\sse\rel$.
\end{lemma}

\begin{proof}
We proceed by induction on $n$.
If there is only one strand with respect to $\vc {A'}n$, say, if $n=1$, there is 
nothing to prove. So, suppose that there are at least two strands.
There are $i,j\in[n]$ and $\ba,\ba'\in\rel$ such that 
$\ba\in A'_1\tm\ldots\tm A'_n$, $\ba'[i]\in A'_i$ and $\ba'[j]\in A_j-A'_j$. 
Let $J\sse[n]$ be the set of all $\ell\in[n]$ with $\ba'[\ell]\in A_\ell-A'_\ell$.
Choose $\ba'$ such that $J$ is minimal. Without loss of generality, 
$J=[s]$ for $s<n$. Set $\bc=\pr_{[n]-J}\ba,\bc'=\pr_{[n]-J}\ba'$ and 
$\bb=\pr_J\ba,\bb'=\pr_J\ba'$.

We show first that these tuples can be chosen such that $\bc=\bc'$.
Let $A'=\pr_{[n]-J}\rel\cap(A'_{s+1}\tm\ldots\tm A'_n)$ and 
$B'=\pr_J\rel\cap(A'_1\tm\ldots\tm A'_s)$.
By Lemma~\ref{lem:connectivity} $A'$ is an as-component of 
$\pr_{[n]-J}\rel$ and $B'$ is an as-component of $\pr_J\rel$. 
Since $\bc,\bc'\in A'$, there is a path 
$\bc=\bc_1,\bc_2\zd \bc_k=\bc'$. Choose some 
$\vc\bb k\in\pr_J\rel$ such that $\bb_1=\bb$, $\bb_k=\bb'$ and $(\bb_i,\bc_i)\in\rel$ for $i\in[t]$.
There is $i$ such that $\bb_i\in B'$, but $\bb_{i+1}\in \pr_J\rel-B'$. Observe that 
$\bb_i\le\bd=\bb_i\bb_{i+1}$, and $\bd\bd'$, $\bd'=h(\bd,\bd,\bb_{i+1})$, is an
affine edge. Therefore $\bd,\bd'\in B'$. Then\\
if $\bc_i\bc_{i+1}$ is semilattice then 
$\cl\bd{\bc_{i+1}}=\cl{\bb_i}{\bc_i}\cdot\cl{\bb_{i+1}}{\bc_{i+1}}$ belongs to $\rel$, or\\
if $\bc_i\bc_{i+1}$ is affine then 
$\cl{\bd'}{\bc_{i+1}}=h\left(\cl{\bb_i}{\bc_i}\cdot\cl{\bb_{i+1}}{\bc_{i+1}},
\cl{\bb_i}{\bc_i}\cdot\cl{\bb_{i+1}}{\bc_{i+1}},\cl{\bb_{i+1}}{\bc_{i+1}}\right)$
 belongs to $\rel$.\\
Either way, $\bc=\bc'$ can be chosen to be $\bc_{i+1}$, and $\bb=\bb_i$ and 
$\bb'$ to be $\bd$ or $\bd'$.

We consider $\rel$ as a subdirect product of $\pr_J$ and $\pr_{[n]-J}\rel$. 
Recall that $\tol_1(\rel)$ denotes the congruence generated by all pairs 
$(\bd,\bd')\in(\pr_J\rel)^2$ that have a common extension $\be\in\pr_{[n]-J}\rel$ 
with $(\bd,\be),(\bd',\be)\in\rel$. By what is already proved $\tol_1(\rel)$ is nontrivial, 
and there are $\bb\in B'$ and $\bb'\not\in B'$ with $(\bb,\bb')\in\tol_1(\rel)$. We
prove that $B'$ is in a $\tol_1(\rel)$-block. For elements $\bb,\bb'$ we take the
ones found in the previous paragraph; there is also $\bc\in A'$ such that
$(\bb,\bc),(\bb',\bc)\in\rel$. Now, if $\al=\tol_1(\rel)$ is 
nontrivial on $B'$, choose $\bd\in B'$ from a different $\al$-block
than $\bb$, and such that $\bd^\al\bb^\al$ is either semilattice or affine. 

First, note that for any $i\in J$ the edge $\bb'[i]\bd[i]$ is either semilattice or 
majority. Indeed, suppose this is not the case. If $\bd[i]\bb'[i]$ is semilattice or 
affine then $\bb'[i]\in A'_i$, a contradiction with the construction. Therefore,
$p(\bb',\bd)=\bd$, while $p(\bb'^\al,\bd^\al)=\bb^\al$, a contradiction again.

To complete the proof it remains to apply the lemma to $\pr_J\rel$ and $\pr_{[n]-J}\rel$.
\end{proof}

\section{Solving conservative CSPs}

Let $\gA$ be a finite class of conservative algebras closed under subalgebras
and retracts. For example, as we noted $\gA$ can be the set of all 
subalgebras of a finite conservative algebra.
In this section we present an algorithm solving $\CSP(\gA)$.
We start with two reductions of the problem.


\subsection{The as-component exclusion reduction}\label{sec:as-exclusion}

The first reduction converts the problem to a number of CSP instances in 
which every domain is an as-component, and then either provides a solution,
or allows to eliminate some elements from some of the original domains.

Let $\cP=(V,\dl,\cC)$ be a $\CSP(\gA)$ instance. 
Choose as-com\-po\-nents $A'_v\sse\dl(v)$ for each $v\in V$
such that for any constraint $\ang{(\vc vn),\rel}$ the set 
$\rel\cap(A'_{v_1}\tm\ldots\tm A'_{v_n})$ is nonempty. We call such a 
collection of as-components a \emph{consistent collection}. A strand of
$\cP$ with respect to $A'_v$, $v\in V$, is a maximal set $W\sse V$
such that for any partition $W_1,W_2$ of $W$ some $w_1\in W_1$,
$w_2\in W_2$ are in the same strand with respect to $A'_{v_1}\zd A'_{v_n}$
of a constraint $\ang{(\vc vn),\rel}\in\cC$. Let $\vc Wk$ be the partion of
$V$ into strands with respect to $A'_v$, $v\in V$. For $i\in[k]$ denote by
$\cP_i$ the problem instance $(W_i,\dl'_i,\cC_i)$, where $\dl'_i:W_i\to\gA$
with $\dl'_i(v)=A'_v$, and 
for each $\ang{(\vc vn),\rel}\in\cC$ we include into $\cC_i$ the constraint
$\ang{(v_{i_1}\zd v_{i_\ell}),\pr_{\{\vc i\ell\}}\rel}$ and $i_j$ is
the positions of $v_j\in W_i$. 

\begin{lemma}\label{lem:strand-decomposition}
If every $\cP_i$ has a solution then $\cP$ has a solution.
\end{lemma}

\begin{proof}
Let $\vf_i$ be a solution of $\cP_i$. Then applying 
Lemma~\ref{lem:rectangularity} to each constraint relation of $\cP$
we conclude that $\vf$ such that $\vf(v)=\vf_i(v)$ whenever $v\in W_i$
is a solution for $\cP$.
\end{proof}

If for some $i\le k$ the problem $\cP_i$ has no solution, then $\cP$ has
no solution $\vf$ with $\vf(v)\in A'_v$ for any $v\in W_i$. Therefore,
$\cP$ can be reduced to a smaller problem $(V,\dl',\cC')$, where
$\dl'(v)=\dl(v)-A'_v$ if $v\in W_i$ and $\dl'(v)=\dl(v)$ otherwise; and 
every constraint relation $\rel$ of $\cP$ is obtained from the corresponding
constraint relation of $\cP$ by restricting it to the new domains.

It remains to show that such a consistent collection of as-components 
always exists, and to demonstrate how it can be found. 

Let $W\sse V$. A \emph{partial consistent collection} on $W$ is 
a collection of as-components $A'_v\sse\dl(v)$ for each $v\in W$
such that for any constraint $\ang{\bs,\rel}$, where $\bs\cap W=(\vc vn)$ 
the set $\pr_{\bs\cap W}\rel\cap(A'_{v_1}\tm\ldots\tm A'_{v_n})$ is 
nonempty.

\begin{prop}\label{pro:consistent-components}
Let $\cP=(V,\dl,\cC)$ be a 3-minimal instance and $W\sse V$. Then any 
partial consistent collection on $W$ can be extended to a consistent collection.
\end{prop}

Observe that Proposition~\ref{pro:consistent-components} implies
that a consistent collection always exists (it suffices to start with 
empty $W$). It also gives a method of finding a consistent collection:
Let $V=\{\vc vn\}$ and choose any as-component $A'_{v_1}$. Then, 
if a partial consistent collection $A'_{v_1}\zd A'_{v_k}$ is chosen,
Proposition~\ref{pro:consistent-components} guarantees that we can
find $A'_{v_{k+1}}$ such that $A'_{v_1}\zd A'_{v_k},A'_{v_{k+1}}$ 
is partial consistent collection.

We start with a statement that is quite similar to 
Proposition~\ref{pro:consistent-components}, but uses relations rather 
than CSP instances. (Partial) consistent collections for relations are 
defined as follows: Let $\rel\le\zA_1\tm\ldots\tm\zA_n$, as-components 
$\vc{A'}n$ is a consistent collection if for any $i,j\in[n]$ the set 
$\pr_{i,j}\rel\cap(A'_i\tm A'_j)$ is non-empty.

\begin{lemma}\label{lem:max-extension}
Let $\rel$ be an ($n$-ary) relation and $I\sse[n]$. For any $\ba\in\pr_I\rel$
such that $\ba[i]$, $i\in I$, belongs to an as-component, there is  $\bb\in\rel$ 
such that $\bb[i]$, $i\in[n]$, belongs to an as-component and $\bb[i]=\ba[i]$ 
for $i\in I$. 
\end{lemma}

\begin{proof}
Consider $\rel$ as a subdirect product of $\rel_1=\pr_I\rel$ and 
$\rel_2=\pr_{[n]-I}\rel$. By Lemma~\ref{lem:connectivity} 
$\ba$ belongs to an as-component of $\rel_1$, and it suffices 
to find $\bb$ in an as-component of $\rel_2$ such that $(\ba,\bb)\in\rel$. 

Let $(\ba,\bb)\in\rel$ for some $\bb\in\rel_2$. If $\bb[i]$ 
belongs to an as-component, we may replace $I$ with $I\cup\{i\}$, so
assume $\bb[i]$ does not belong to an as-component for $i\in[n]-I$.
Take $\bc=(\bc_1,\bc_2)\in\rel$ with $\bc_1\in\rel_1$ and $\bc_2$
from an as-component of $\rel_2$. As $\bb[i]$ is not in any as-component,
$\bb[i]\bc_2[i]$ is a semilattice or majority edge for $i\in[n]-I$. 
Letting $\bd=\bb\cdot\bc_2$ we have that $\bb\bd$ is a semilattice edge 
and $\bd\bc_2$ is a majority edge. Observe that 
$$
\cl{\ba\cdot\bc_1}{\bd}=\cl\ba\bb\cdot\cl{\bc_1}{\bc_2}\in\rel 
\quad\hbox{and}\quad
\cl{p(\ba\cdot\bc_1,\bc_1)}{\bc_2}=
p\left(\cl{\ba\cdot\bc_1}{\bd},\cl{\bc_1}{\bc_2}\right)\in\rel,
$$
and that $p(\ba\cdot\bc_1,\bc_1)$ belongs to the same 
as-component as $\ba$. Thus, by Lemma~\ref{lem:connectivity}
$(\ba,\bc_3)\in\rel$ for some $\bc_3$ from the same as-component 
as $\bc_2$.
\end{proof}

Lemma~\ref{lem:max-extension} implies that for any relation 
there is a consistent collection. Indeed, if $\ba\in\rel$ is such that 
$\ba[i]$ belongs to an as-component $A'_i$, then $\vc{A'}n$
is a consistent collection.

\begin{lemma}\label{lem:CRT}
Let $\vc{A'}n$ be a consistent collection for an $n$-ary relation $\rel$.
Then $(A'_1\tm\ldots\tm A'_n)\cap\rel\ne\eps$.
\end{lemma}

\begin{proof}
We prove by induction that for any $I\sse [n]$ there is $\ba\in\rel$
such that $\ba[i]\in A'_i$ for $i\in I$. Since $\vc{A'}n$ is a consistent collection,
the statement is true for any $I$ with $|I|\le2$. Suppose it is true for 
any $J\sse [n]$ such that $|J|<|I|$. Without loss of generality assume 
$1,2,3\in I$. Let $J_1= I-\{1\}$, $J_2=I-\{2\}$, $J_3=I-\{3\}$, and let
$\ba_1,\ba_2,\ba_3\in\rel$ such that $\ba_j[i]\in A'_i$ for all $i\in J_j$.
If one of $\ba_j[j]\in A'_j$, $j\in\{1,2,3\}$, then we are done; assume this
is not the case. By Lemma~\ref{lem:connectivity} 
$(A'_1\tm A'_3\tm\ldots\tm A'_n)\cap\pr_{\{1,3\zd n\}}\rel$ and 
$(A'_1\tm A'_2\tm A'_4\tm\ldots\tm A'_n)\cap\pr_{\{1,2,4\zd n\}}\rel$ are 
subdirect products of $A'_1,A'_3\zd A'_n$ and $A'_1,A'_2,A'_4\zd A'_n$, respectively.
Therefore $\ba_2,\ba_3$ can be chosen so that $\ba_1[3]=\ba_2[3]$ 
and $\ba_1[2]=\ba_3[2]$. While $\ba_1[i],\ba_2[i]\in A'_i$ for all 
$i\in\{3\zd n\}$, $\ba_2[1]\in A'_1$, and $\ba_1[1]\not\in A'_1$, by
Lemma~\ref{lem:rectangularity} 
$$
(A'_1\tm A'_3\tm\ldots\tm A'_n)\cap\pr_{\{1,3\zd n\}}\rel=
A'_1\tm\left[(A'_3\tm\ldots\tm A'_n)\cap\pr_{\{3\zd n\}}\rel\right].
$$
Hence, $\ba_2$ can be assumed such that $\ba_2[1]=\ba_3[1]$.

If $\ba_j[j]\ba_k[j]$ is a semilattice edge for some $j,k\in\{1,2,3\}$ 
then the tuple $\ba_j\ba_k$ satisfies the required conditions. It remains 
to consider the case when $\ba_j[j]\ba_k[j]$ is a majority edge for any
$j,k\in\{1,2,3\}$. Consider $\bb=g(\ba_1,\ba_2,\ba_3)$. As 
$\ba_1[i],\ba_2[i],\ba_3[i]\in A'_i$ for $i\in\{4\zd n\}$, we have 
$\bb[i]\in A'_i$ in this case. Since $\ba_2[1]=\ba_3[1]$ and 
$\ba_1[1]\ba_2[1]$ is a majority edge, $\bb[1]=\ba_2[1]$. 
Similarly, $\bb[2]=\ba_1[2]\in A'_2$ and $\bb[3]=\ba_1[3]\in A'_3$.
\end{proof}

\begin{corollary}\label{cor:CRT}
Let $\rel\le\tc\zA n$, and let $\vc{A'}{n-1}$ be a partial consistent  
collection, $A'_i\sse A_i$. Then it can be extended to a consistent 
collection for $\rel$.
\end{corollary}

\begin{proof}
By Lemma~\ref{lem:CRT} there is $\ba\in(\tc{A'}{n-1})\cap\pr_{[n-1]}\rel$.
Therefore, by Lemma~\ref{lem:max-extension} $(\ba,a)\in\rel$ for some
$a$ from an as-component of $\zA_n$.
\end{proof}

By $\relo_{v,w},\relo_{u,v,w}$ we denote the sets of partial solutions of $\cP$ 
on $\{v,w\}$ and $\{u,v,w\}$, respectively.

\begin{proof}(of Proposition~\ref{pro:consistent-components})
The proof we give here is a modification of the proof
of Theorem~3.5 from \cite{Jeavons98:consist}.

Suppose $\cP=(V,\dl,\cC)$ is a minimal instance that does not satisfy 
the conclusion of the proposition. Since we assume $\cP$ 3-minimal, 
$|V|>3$. Pick $v\in V$; our 
assumption implies that $\cP_{V-\{v\}}$ satisfies the conclusion 
of the proposition, but there is a consistent collection 
$\{A'_w\sse\dl(w)\mid w\in W=V-\{v\}\}$ such that it cannot be extended to 
a consistent collection including some $A'_v\sse\dl(v)$.

Let $\cC=\{\ang{\bs_1,\rel_1}\zd
\ang{\bs_q,\rel_q}\}$. To obtain the desired contradiction we shall
construct a problem $\cP'$ which also has $q$ constraints, with the
same constraint relations, but with different constraint scopes.

We define the set of variables of $\cP'$ to be the union of $\{v'\}$
and $q$ disjoint copies $\vc Wq$ of $W$, where $W_i=\{w^i_1\zd
w^i_k\}$.  
Now, for each $i\in[q]$, we define a mapping $f_i\colon W\to W_i$ by
setting $f_i(w_j)=w^i_j$, and extend each $f_i$ to $v$ by setting
$f_i(v)=v'$. The set of constraints of $\cP'$ is then defined as
$$
\{\ang{f_1(\bs_1),\rel_1}\zd\ang{f_q(\bs_q),\rel_q}\}.
$$

Then let the $q\cdot k$-ary relation $\rel$ be defined as follows
\vspace*{-1.5mm}
\begin{eqnarray*}
\rel &=& \{(\vf(f_1(w_1))\zd\vf(f_1(w_k))\zd
\vf(f_q(w_1))\zd\vf(f_q(w_k)),\vf(v')\mid\\ 
& & \hbox{$\vf$ is a solution to $\cP'$}\}.
\end{eqnarray*}
\vspace*{-5mm}

\noindent
The collection $A'_{v_1}\zd A'_{v_k}\zd  A'_{v_1}\zd A'_{v_k}$
cannot be extended to a consistent collection for $\rel$, since $A'_{v_1}\zd A'_{v_k}$
cannot be extended to a consistent collection for $\cP$. 
However, we shall show that $\rel$
satisfies the conditions of Lemma~\ref{cor:CRT}, 
and thus derive a contradiction.

For any pair of
indices $w^{i_1}_{j_1},w^{i_2}_{j_2}$, we claim that 
$(A'_{j_1}\tm A'_{j_2})\cap\pr_{\{w^{i_1}_{j_1},w^{i_2}_{j_2}\}}\rel\ne\eps$. 
Since $\cP$ is 3-minimal any tuple $(a,b)\in(A'_{j_1}\tm A'_{j_2})\cap\relo_{w_{j_1}w_{j_2}}$
can be extended to a solution $(a,b,c)\in\relo_{w_{j_1}w_{j_2},v}$. 
Furthermore, for this solution, we can
construct a corresponding solution, $\vf$, to $\cP'$, such that
$\vf(f_{i_1}(w_{j_1}))=\vf_W(w_{j_1})$. Indeed, for any constraint 
$\ang{\bs_j,\rel_j}$, this partial solution can be extended to a tuple $\ba$ 
from $\rel_j$. Then we assign values to $f_j(w_1)\zd f_j(w_k)$ accordingly 
to $\ba$ (the variable that are not in the constraint scope $f_j(\bs_j)$ 
can be assigned values arbitrarily). 

Now, by Corollary~\ref{cor:CRT} we get a contradiction.
\end{proof}

\subsection{Maroti's reduction}\label{sec:maroti}

Reductions of the second type will be applied to instances, in which all the domains
are as-components, but some of them contain semilattice edges. We will
call such instances \emph{semilattice free}.

Maroti in \cite{Maroti10:tree} suggested a reduction for CSPs that are invariant
under a certain binary operation. Let $\gA$ be a class of finite algebras of 
similar type closed under subalgebras. 
Suppose that $\gA$ has a term operation $f$ satisfying the 
following conditions for some $\zA\in\gA$:
\begin{enumerate}
\item
$f(x,f(x,y))=f(x,y)$ for any $x,y\in\zA$;
\item
$\gA$ is closed under retracts via unary polynomials $f(a,x), f(x,a)$;
\item
for each $a\in\zA$ the mapping $x\mapsto f(a,x)$ is not surjective;
\item
the set $C$ of $a\in\zA$ such that $x\mapsto f(x,a)$ is surjective 
generates a proper subalgebra of $\zA$.
\end{enumerate}
Then $\CSP(\gA)$ is polynomial time reducible to $\CSP(\gA-\{\zA\})$.

As is easily seen, the operation $\cdot$ of a class $\gA$ of conservative algebras of
closed under subalgebras and any $\zA\in\gA$
satisfies conditions (1),(2). If the operation $a\cdot x$ is surjective for some $a$,
then $a\le x$  for all $x\in\zA$. Therefore the only case when condition 
(3) is not satisfied is when $\zA$ has such a minimal element. Finally, 
condition (4) is satisfied whenever $\zA$ is not semilattice free. 

We apply Maroti's reduction only in the case when every domain of the instance
is either semilattice free, or is an as-component. In this situation this 
reduction can be slightly modified. More precisely, we will apply it to all 
semilattice free domains rather than just one. Below we explain the reduction,
and the modifications required. The reduction uses 3 types of constructions.

Let $\cP=(V,\dl,\cC)$
be an instance of $\CSP(\gA)$ and $p_v\colon\dl(v)\to\dl(v)$, $v\in V$. 
Mappings $p_v$, $v\in V$, are said to be \emph{consistent} if for any
$\ang{\bs,\rel}\in\cC$, $\bs=(\vc vk)$, and any tuple $\ba\in\rel$ the
tuple $(p_{v_1}(\ba[1])\zd p_{v_k}(\ba[k]))$ belongs to $\rel$. 
Mappings $p_v$ are called \emph{permutational} if all of them are 
permutations, they are called \emph{idempotent} if all of them are idempotent.
For consistent idempotent mappings $p_v$ by $p(\cP)$ we denote the
\emph{retraction} of $\cP$, that is, $\cP$ restricted to the images of 
$p_v$. As is easily seen (see  \cite{Maroti10:tree}), in this case $\cP$ has
a solution if and only if $p(\cP)$ has. Also, if $p_v$ are consistent
non-permutational maps, then there are consistent idempotent maps
$p'_v$ of $\cP$ obtained by iterating $p_v$. 

The next construction uses a binary idempotent operation $\cdot$
satisfying the identity $x\cdot(x\cdot y)=x\cdot y$. Then $t(\cP)$
denotes the instance $(V',\dl',\cC')$ where
\begin{itemize}
\item
$V'=\{(v,b)\mid v\in V, b\in\dl(v)\}$ is the set of variables;
\item 
the domains are defined by the rule $\dl'(v,b)=b\cdot\dl(v)=
\{b\cdot x\mid x\in\dl(v)\}$;
\item
$\cC'$ contains constraints of two types:\\
first, for each $v\in V$, it contains the constraint $\ang{\bs_v,\rel_v}$
where $\bs_v=((v,b_1)\zd(v,b_k))$ for some enumeration $\vc bk$
of elements of $\dl(v)$, and $\rel_v=\{(b_1\cdot c\zd b_k\cdot c)\mid
c\in\dl(v)\}$;\\
second, for every $C=\ang{\bs,\rel}\in\cC$, $\bs=(\vc vk)$, and 
$\ba\in\rel$ there is constraint $D_{C,\ba}=\ang{\bs_{C,\ba},\rel_{C,\ba}}$
given by $\bs_{C,\ba}=((v_1,\ba[1])\zd(v_k,\ba[k]))$ and 
$\rel_{C,\ba}=\{\ba\cdot\bx\mid\bx\in\rel\}$.
\end{itemize}
The important property of the problem $t(\cP)$ is that if it has a solution
$\vf$ then mappings $p_v$, $v\in V$, given by $p_v(b)=\vf(v,b)$
are consistent. If $t(\cP)$ does not have a solution, $\cP$ also does not
have a solution (see \cite{Maroti10:tree})

We describe the last construction used in the reduction for conservative algebras only. Let $B_v$
be the set of all $b\in\dl(v)$ such that $ab$ is a semilattice edge for
no $a\in\dl(v)$. For every such $b$ the mapping $x\cdot b$ is 
injective, while for any other $b$ it is not. Then let $c(\cP)$
denote the restriction of $\cP$ to the sets $B_v$.

The reduction then goes as follows. First, solve $c(\cP)$. If it has a solution, it is
also a solution of $\cP$, so assume $c(\cP)$ has no solution. If $t(\cP)$
has a solution that is not permutational, then $\cP$ has consistent
non-permutational mappings, $p_v$, that can be assumed idempotent.
In this case $\cP$ has a solution if and only if $p(\cP)$ has, and can
be replaced with this smaller problem, as $\summ(p(\cP))<\summ(\cP)$.
It remains to consider the case when $p(\cP)$ has no solution that gives
rise to non-permutational mappings.

In this case, as $c(\cP)$ has no solution, for any solution $\vf$ of $\cP$,
there is $v\in V$ such that $\vf(v)=b\not\in B_v$. Then for each variable
$w\in V$ and every $\dl(w)- B_w$ we create the instance $t(\cP)$ with an
additional unary constraints $\ang{(w,b),(b\cdot d)}$, $b\in\dl(w)$. This
implies that for any consistent maps $p_v$ that arise from a solution 
to such instance, $p_w(b)=b\cdot d$, and therefore, they are not 
permutational. If there is such a non-permitational collection of 
consistent mappings, we replace $\cP$ with $p(\cP)$; otherwise we conclude
that $\cP$ has no solution.


\subsection{The algorithm and its running time}

Consider an instance $\cP=(V,\dl,\cC)$ of $\CSP(\gA)$.
Recall that it is called semilattice free if none of $\cG(\dl(v))$ 
contains a semilattice edge. Our algorithm works recursively reducing 
the domains so that eventually we obtain a semilattice free instance.

First, we show how to solve semilattice free instances. Every edge of $\cG(\dl(v))$,
$v\in V$, in this case is either majority or affine. Therefore for any $v\in V$ and 
any $a,b\in\dl(v)$ the operation 
$m(x,y,z)=h(g(x,y,z),g(y,z,x),g(z,x,y))$ is a majority operation if $ab$ 
is a majority edge, and is an affine operation if $ab$ is an affine edge. 
Thus $m$ satisfies the conditions of a \emph{generalized 
majority-minority} operation, and can be solved by the algorithm
from \cite{Dalmau06:majority-minority}.

If $\cP$ is not semilattice free, but every domain is an as-component, we
apply Maroti's reduction, as described in Section~\ref{sec:maroti}. This
reduction repeatedly reduces the problem to a smaller one, $p(\cP)$, by finding 
consistent maps $p$, and either discovers that $\cP$ does not have a solution 
or produces a problem which is semilattice free or has a proper as-component.
It also makes recursion calls with instances $t(\cP)$ and $c(\cP)$, each of which 
is either semilattice free or has a domain with a least element and therefore
with a proper as-component. 

Finally, if $\cP$ has a domain with a proper as-component, we apply the as-component 
exclusion reduction as described in Section~\ref{sec:as-exclusion}, and either find
a solution or reduce some of the domains. This reduction makes recursive calls with 
instances in which every domain is an as-component.

The correctness of this algorithm follows from the previous sections, 
\cite{Maroti10:tree}, and \cite{Dalmau06:majority-minority}. Therefore,
it remains to prove that the algorithm is polynomial time.

\begin{prop}\label{pro:complexity}
The algorithm is polynomial time in the size of $\cP$.
\end{prop}

Solving semilattice free instances is polynomial time by \cite{Dalmau06:majority-minority}.
We consider the recursion tree generated by the algorithm. It is easy to
see that at every node of the tree the amount of work done by the
algorithm is bounded by a polynomial, so is the number of recursive calls.
Therefore it suffices to show that the depth of recursion is bounded by a 
constant.

Let $\lev(\cP)$ for an instance $\cP$ of $\CSP(\gA)$ denote the maximal 
size of a semilattice non-free domain of $\cP$.
The following lemma is straightforward.

\begin{lemma}\label{lem:maroti-reduction}
Let $\cP=(V,\dl,\cC)$ be an instance of $\CSP(\gA)$ such that all $\dl(v)$ are 
as-components (and therefore do not have a least element). Let also $p_v$, 
$v\in V$, be consistent maps for $\cP$. Then
$p(\cP),t(\cP),c(\cP)$ are instances of $\CSP(\gA)$, and 
$\lev(t(\cP)),\lev(c(\cP))<\lev(\cP)$;
\end{lemma}

We use the following observation:
\begin{quote}
Suppose there is a constant $c$ such that for any problems $\cP'$ and $\cP''$ such that
$\cP''$ is a successor of $\cP'$ in the recursion tree and the length of the path
from $\cP'$ to $\cP''$ is at least $c$, then $\lev(\cP'')<\lev(\cP')$.
Then the recursion tree
has depth at most $c\cdot k$ where $k$ is the maximal size of a semilattice non-free
algebra in $\gA$.
\end{quote}

We show that the algorithm satisfies the condition above for $c=2$.  Let 
$\cP'$ be the problem being solved at some node of the recursion tree. 
Suppose first that all the domains of $\cP'$ are semilattice-free. Then $\cP'$ 
has no successors and there is nothing to prove. Next, suppose that some 
domain is not an as-component. Then every child of $\cP'$ is of the form
$\cP_{I_j}$ for some strand $I_j$. Every domain in a problem like this is
an as-component. Note, however, that the size of at least some domains 
may not decrease at this step, if those domains are already as-components.
Finally, suppose that all domains of $\cP'$ are as-components. Then every 
child of $\cP'$ has the form $c(\cP')$, $t(\cP')$, or $\cP'_{v,d}=
t(\cP')\cup\{\ang{(v,d),d}\}$. By Lemma~\ref{lem:maroti-reduction} the maximal 
size of semilattice non-free domain of each of these problems is strictly less 
than that of $\cP'$.


\begin{thebibliography}{10}

\bibitem{Barto11:conservative}
L.~Barto.
\newblock The dichotomy for conservative constraint satisfaction problems
  revisited.
\newblock In {\em LICS}, pages 301--310, 2011.

\bibitem{Bulatov11:conservative}
A.~Bulatov.
\newblock Complexity of conservative constraint satisfaction problems.
\newblock {\em ACM Trans. Comput. Log.}, 12(4):24, 2011.

\bibitem{Bulatov02:3-element}
A. Bulatov.
\newblock A dichotomy theorem for constraints on a three-element set.
\newblock In {\em FOCS}, pages 649--658, 2002. 

\bibitem{Bulatov03:conservative}
A. Bulatov.
\newblock Tractable conservative constraint satisfaction problems.
\newblock In {\em LICS}, pages 321--330, 2003. 

\bibitem{Bulatov03:multisorted}
A. Bulatov and P.G. Jeavons.
\newblock An algebraic approach to multi-sorted constraits.
\newblock In {\em CP}, pages 197--202, 2003.

\bibitem{Bulatov06:3-element}
A.~Bulatov.
\newblock A dichotomy theorem for constraint satisfaction problems on a
  3-element set.
\newblock {\em J. ACM}, 53(1):66--120, 2006.

\bibitem{Bulatov05:classifying}
A.~Bulatov, P.G.~Jeavons, and A.~Krokhin.
\newblock Classifying the complexity of constraints using finite algebras.
\newblock {\em SIAM J. Comput.}, 34(3):720--742, 2005.

\bibitem{Burris81:universal}
S.~Burris and H.P. Sankappanavar.
\newblock {\em A course in universal algebra}, volume~78 of {\em Graduate Texts
  in Mathematics}.
\newblock Springer-Verlag, 1981.

\bibitem{Dalmau06:majority-minority}
V.~Dalmau.
\newblock Generalized majority-minority operations are tractable.
\newblock {\em Logical Methods in Comput. Sci.}, 2(4), 2006.

\bibitem{Dechter03:processing}
R.~Dechter.
\newblock {\em Constraint processing}.
\newblock Morgan Kaufmann Publishers, 2003.

\bibitem{Feder98:list}
T.~Feder and P.~Hell.
\newblock List homomorphisms to reflexive graphs.
\newblock {\em J. {C}omb. {T}heory {B}}, 72:236--250, 1998.

\bibitem{Feder99:list}
T.~Feder, P.~Hell, and J.~Huang.
\newblock List homomorphisms and circular arc graphs.
\newblock {\em Combinatorica}, 19:487--505, 1999.

\bibitem{Feder99:bi-arc}
T.~Feder, P.~Hell, and J.~Huang.
\newblock Bi-arc graphs, and the complexity of list homomorphisms.
\newblock {\em J. {G}raph {T}heory}, 42(1):61--80, 2003.

\bibitem{Feder98:monotone}
T.~Feder and M.Y. Vardi.
\newblock The computational structure of monotone monadic {SNP} and constraint
  satisfaction: A study through datalog and group theory.
\newblock {\em {SIAM} J. Comput.}, 28:57--104, 1998.

\bibitem{Hell90:h-coloring}
P.~Hell and J.~Ne\v{s}et\v{r}il.
\newblock On the complexity of {$H$}-coloring.
\newblock {\em J. Comb. Theory B}, 48:92--110, 1990.

\bibitem{Jeavons98:algebraic}
P.G. Jeavons.
\newblock On the algebraic structure of combinatorial problems.
\newblock {\em Theor. Comput. Sci.}, 200:185--204, 1998.

\bibitem{Jeavons98:consist}
P.G. Jeavons, D.A. Cohen, and M.C. Cooper.
\newblock Constraints, consistency and closure.
\newblock {\em Artificial Intelligence}, 101(1-2):251--265, 1998.

\bibitem{Jeavons97:closure}
P.G. Jeavons, D.A. Cohen, and M.~Gyssens.
\newblock Closure properties of constraints.
\newblock {\em J. ACM}, 44:527--548, 1997.

\bibitem{Kratochvil94:algorithmic}
J.~Kratochvil and Z.~Tuza.
\newblock Algorithmic complexity of list colorings.
\newblock {\em Discr. {A}ppl. {M}ath.}, 50:297--302, 1994.

\bibitem{Maroti10:tree}
M.~Mar\'oti.
\newblock Tree on top of {M}altsev.
\newblock manuscript, 2010.

\bibitem{Schaefer78:complexity}
T.J. Schaefer.
\newblock The complexity of satisfiability problems.
\newblock In {\em STOC}, pages 216--226, 1978.

\end{thebibliography}
\end{document}